\documentclass[a4paper,USenglish,cleveref,numberwithinsect]{lipics-v2021}
\pdfoutput=1
\hideLIPIcs
\nolinenumbers

\usepackage[capitalise]{cleveref}
\usepackage{booktabs}
\usepackage{cite}

\newcommand{\expdensesr}{4/3}
\newcommand{\expdensefield}{1.157}
\newcommand{\explbsr}{1.333}
\newcommand{\explbfield}{1.156}
\newcommand{\expprevsr}{1.927}
\newcommand{\expprevfield}{1.907}
\newcommand{\expnewsr}{1.867}
\newcommand{\expnewfield}{1.832}
\newcommand{\eps}{\varepsilon}
\newcommand{\alg}{\mathcal{A}}

\newcommand{\TT}{\mathcal{T}}
\newcommand{\PP}{\mathcal{P}}
\newcommand{\TTall}{\hat{\TT}}
\newcommand{\US}{\mathsf{US}}
\newcommand{\RS}{\mathsf{RS}}
\newcommand{\CS}{\mathsf{CS}}
\newcommand{\BD}{\mathsf{BD}}
\newcommand{\AS}{\mathsf{AS}}
\newcommand{\GM}{\mathsf{GM}}
\newcommand{\XX}{\mathsf{X}}
\newcommand{\YY}{\mathsf{Y}}
\newcommand{\ZZ}{\mathsf{Z}}
\newcommand{\OR}{\mathsf{OR}}
\newcommand{\bb}{\mathsf{b}}

\newcommand{\GG}{\mathcal{G}}
\DeclareMathOperator{\out}{out}

\title{Low-Bandwidth Matrix Multiplication: Faster Algorithms and More General Forms of Sparsity}
\titlerunning{Low-Bandwidth Matrix Multiplication}
\authorrunning{C. Gupta, J.\,H. Korhonen, J. Studen\'y, J. Suomela and H. Vahidi}

\author{Chetan Gupta}{IIT Roorkee, India}{chetan.gupta@cs.iitr.ac.in}{https://orcid.org/0000-0002-0727-160X}{}

\author{Janne H. Korhonen}{Finland} {janne.h.korhonen@gmail.com}{https://orcid.org/0009-0000-5494-1218}{}

\author{Jan Studen\'y}{Aalto University, Finland} {jan.studeny@aalto.fi}{https://orcid.org/0000-0002-9887-5192}{This work was supported in part by the Research Council of Finland, Grants 321901, and 333837.}

\author{Jukka Suomela}{Aalto University, Finland} {jukka.suomela@aalto.fi}{https://orcid.org/0000-0001-6117-8089}{}

\author{Hossein Vahidi}{Aalto University, Finland} {hossein.vahidi@aalto.fi}{https://orcid.org/0000-0002-0040-1213}{This work was supported in part by the Research Council of Finland, Grant 333837.}

\acknowledgements{We are grateful to Juho Hirvonen and Jara Uitto for several fruitful discussions, and to the anonymous reviewers for their helpful feedback on prior version of this work.}



\begin{document}
	\maketitle
	
	\begin{abstract}
		In prior work, Gupta et al.\ (SPAA 2022) presented a distributed algorithm for multiplying sparse $n \times n$ matrices, using $n$ computers. They assumed that the input matrices are \emph{uniformly sparse}---there are at most $d$ non-zeros in each row and column---and the task is to compute a uniformly sparse part of the product matrix. The sparsity structure is globally known in advance (this is the \emph{supported} setting). As input, each computer receives one row of each input matrix, and each computer needs to output one row of the product matrix. In each communication round each computer can send and receive one $O(\log n)$-bit message. Their algorithm solves this task in $O(d^{1.907})$ rounds, while the trivial bound is $O(d^2)$.
		
		We improve on the prior work in two dimensions: First, we show that we can solve the same task faster, in only $O(d^{1.832})$ rounds. Second, we explore what happens when matrices are not uniformly sparse. We consider the following alternative notions of sparsity: row-sparse matrices (at most $d$ non-zeros per row), column-sparse matrices, matrices with bounded degeneracy (we can recursively delete a row or column with at most $d$ non-zeros), average-sparse matrices (at most $dn$ non-zeros in total), and general matrices.

		We present a near-complete classification of the complexity of matrix multiplication for all combinations of these notions of sparsity. We show that almost all cases fall in one of these classes:
		\begin{enumerate}
			\item We have an upper bound of $O(d^{1.832})$ rounds for computing $X = AB$. An example is the case where $X$ and $A$ are uniformly sparse but $B$ is average-sparse; this is a generalization of the prior work beyond the uniformly-sparse case.
			\item We have a lower bound of $\Omega(\log n)$ and an upper bound of $O(d^2 + \log n)$ rounds for computing $X = AB$. An example is the case where $X$, $A$, and $B$ have bounded degeneracy.
			\item We have a lower bound of $n^{\Omega(1)}$ rounds for computing $X = AB$. An example is the case where $X$ and $A$ have bounded degeneracy but $B$ is a general matrix.
			\item We have a conditional lower bound: a fast algorithm for computing $X = AB$ would imply major improvements in \emph{dense} matrix multiplication. An example is the case where $X$, $A$, and $B$ are average-sparse matrices.
		\end{enumerate}
		Our work highlights the role that bounded degeneracy has in the context of distributed matrix multiplication: it is a natural class of sparse matrices that is restrictive enough to admit very fast algorithms---much faster than what we can expect for the average-sparse case.
	\end{abstract}
	
	\maketitle
	
	\section{Introduction}
	
	In this work we study sparse matrix multiplication in a distributed setting, for the low-bandwidth model (which is closely related to the node-capacitated clique model). Our work improves over the prior work by \cite{gupta-2022-sparse} in two ways: we present a faster algorithm, and we study more general forms of sparsity.
	
	\subsection{Setting and prior work}
	
	The task is to compute the matrix product $X = AB$ for $n \times n$ matrices using a network of $n$ computers. Initially each computer holds its own part of $A$ and $B$. Computation proceeds in rounds, and in one round each computer can send one $O(\log n)$-bit message to another computer and receive one message from another computer; we will assume that the elements of $A$, $B$, and $X$ fit in one message. Eventually each computer has to report its own part of the product matrix $X$. How many communication rounds are needed to solve the task?
	
	The trivial solution for \emph{dense} matrices takes $O(n^2)$ rounds: everyone sends all information to computer number $1$, which solves the task locally and then distributes the solution to other computers. However, we can do better: for matrix multiplication over semirings there is an algorithm that runs in $O(n^{\expdensesr})$ rounds, and for matrix multiplication over fields there is an algorithm that runs in $O(n^{2-2/\omega})$ rounds, where $\omega$ is the exponent of centralized matrix multiplication \cite{alg-method-congest-cliq2019}. By plugging in the latest value of $\omega < 2.371552$ \cite{vassilevska-williams-2024-omega}, we obtain $O(n^{\expdensefield})$ rounds.
	
	The key question is how much better we can do when we multiply \emph{sparse} matrices. Let us first look at the case of \emph{uniformly sparse} matrices, with at most $d$ nonzero elements in each row and column. We also assume that we are only interested in a uniformly sparse part of the product matrix $X$. Now it is natural to assume that each computer initially holds one row of $A$ and one row of $B$, and it needs to know one row of $X$. We assume the \emph{supported} setting: the sparsity structures of $A$, $B$, and $X$ are known in advance, while the values of the nonzero elements are revealed at run time. There are two algorithms from prior work that are applicable in this setting: for moderately large values of $d$ we can use the algorithm by \cite{censor2021fast}, which runs in $O(dn^{1/3})$ rounds, while for small values of $d$ the fastest algorithm is due to \cite{gupta-2022-sparse}, and the round complexity is $O(d^{\expprevsr})$ for semirings and $O(d^{\expprevfield})$ for fields; see \cref{tab:summary} for an overview.
	
	\begin{table}[b]
		\centering
		\caption{Complexity of distributed sparse matrix multiplication}\label{tab:summary}
		\begin{tabular}{@{}lll@{}}
			\toprule
			Semirings & Fields & Reference \\
			\midrule
			$O(n^2)$ & $O(n^2)$ & trivial \\
			$O(n^{\expdensesr})$ & $O(n^{\expdensefield})$ & \cite{vassilevska-williams-2024-omega,alg-method-congest-cliq2019} \\
			$O(dn^{1/3})$ & $O(dn^{1/3})$ & \cite{censor2021fast} \\
			$O(d^2)$ & $O(d^2)$ & trivial, \cite{gupta-2022-sparse} \\
			$O(d^{\expprevsr})$ & $O(d^{\expprevfield})$ & \cite{gupta-2022-sparse} \\
			$O(d^{\expnewsr})$ & $O(d^{\expnewfield})$ & this work, \cref{thm:us-us-as} \\
			\bottomrule
		\end{tabular}
	\end{table}
	
	\subsection{Contribution 1: faster algorithm}\label{ssec:intro-contrib-1}
	
	Our first contribution is improvements in the running time: we design a faster algorithm that solves the case of semirings in $O(d^{\expnewsr})$ rounds and the case of fields in $O(d^{\expnewfield})$ rounds.
	
	While there are no non-trivial unconditional lower bounds, we point out that $O(d^{\explbsr})$ for semirings or $O(d^{\explbfield})$ for fields would imply major breakthroughs for dense matrix multiplication (by simply plugging in $d = n$). In the following figure, we illustrate the progress we make towards these milestones:
	\begin{center}
		\includegraphics[page=1,scale=0.9]{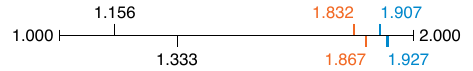}
	\end{center}
	
	The algorithm from prior work \cite{gupta-2022-sparse} is based on the idea of \emph{processing triangles}. Here a \emph{triangle} is a triple $\{i,j,k\}$ such that $A_{ij}$ and $B_{jk}$ are nonzeros and $X_{ik}$ is one of the elements of interest, and we say that we \emph{process} the triangle if we add the product $A_{ij}B_{jk}$ to the sum $X_{ik}$. Once all such triangles are processed, we have computed all elements of interest in the product matrix.
	There are two strategies that \cite{gupta-2022-sparse} uses for processing triangles:
	\begin{enumerate}
		\item If there are many triangles in total, we can find a dense cluster of them, we can interpret such a cluster as a tiny instance of dense matrix multiplication, and we can apply dense matrix multiplication algorithms to batch-process such clusters.
		\item If there are few triangles in total, we can afford to process the triangles one by one.
	\end{enumerate}
	
	Our improvements are in the second phase. Trivially, if we have a uniformly sparse instance, there are at most $d^2n$ triangles in total, and we can process them in a naive manner in $O(d^2)$ rounds. The second phase of \cite{gupta-2022-sparse} shows that if the number of triangles is $d^{2-\eps} n$, one can process them in $O(d^{2-\eps/2})$ rounds. Here the key obstacle is the expression $\eps/2$ in the exponent: even if we eliminate many triangles in the first phase, the second phase is still relatively expensive. We design a new algorithm that is able to process $d^{2-\eps} n$ triangles in $O(d^{2-\eps})$ rounds. By combining this with the general strategy of \cite{gupta-2022-sparse} and optimizing the trade-off between using time in the first phase vs.\ using time in the second phase, we obtain the new round complexities $O(d^{\expnewsr})$ and $O(d^{\expnewfield})$---the difference between semirings and fields is in the dense matrix multiplication routine that we use in the first phase.
	
	\subsection{Contribution 2: beyond uniform sparsity}\label{ssec:intro-contrib-2}
	
	So far we have followed the lead of \cite{gupta-2022-sparse} and discussed \emph{uniformly sparse} matrices (i.e., each row and column has at most $d$ nonzeros). This is a rather restrictive notion of sparsity---can we extend prior work to more general notions of sparsity?
	
	Perhaps the most natural candidate to consider would be \emph{average sparsity}: we would simply assume that the total number of nonzeros in $A$ and $B$ and the total number of elements of interest in $X$ is at most $dn$. Unfortunately, it turns out that this notion of sparsity is way too general, and we cannot expect fast algorithms for this case: in \cref{thm:AS-AS-AS-hard} we show that if we could solve average-sparse matrix multiplication for $n\times n$ matrices using $n$ computers in time that is independent of $n$ or only mildly depends on $n$ (say, polylogarithmic in $n$), it would imply major breakthroughs in our understanding of dense matrix multiplication.
	
	In summary, in order to have any real hope of pushing the complexity of sparse matrix multiplication down to $O(d^{\expnewfield})$ or even something more modest like $O(d^2 + \log n)$, we need to explore some \emph{intermediate notions of sparsity}. In this work we make use of the following notions of sparsity:
	\begin{itemize}
		\item $\US(d)$ = uniformly sparse: at most $d$ nonzeros per row and column.
		\item $\RS(d)$ = row-sparse: at most $d$ nonzeros per row.
		\item $\CS(d)$ = column-sparse: at most $d$ nonzeros per column.
		\item $\BD(d)$ = bounded degeneracy: we can recursively eliminate the matrix so that at each step we delete a row or column with at most $d$ nonzeros.
		\item $\AS(d)$ = average-sparse: at most $dn$ nonzeros in total.
		\item $\GM$ = general matrices.
	\end{itemize}
	We will omit $(d)$ when it is clear from the context, and we will use e.g.\ expressions such as $\US \times \BD = \AS$ to refer to the task of computing $AB = X$ such that $A \in \US(d)$, $B \in \BD(d)$, and $\hat{X} \in \AS(d)$; here we write $\hat{X}$ for the matrix that indicates which elements of the product $X$ we are interested in.
	
	Family $\BD$ may at first look rather unusual in the context of linear algebra. However, if we interpret a matrix $A$ as a bipartite graph $G$ (with an edge $\{i,j\}$ whenever $A_{ij}$ is nonzero), then $A \in \BD(d)$ corresponds to the familiar graph-theoretic notion that $G$ is $d$-degenerate. This is a widely used notion of sparsity in the graph-theoretic setting, and also closely connected with other notions of sparsity such as bounded arboricity.
	
	We also point out that any matrix $A \in \BD(d)$ can be written as a sum $A = X + Y$ such that $X \in RS(d)$ and $Y \in CS(d)$; to see this, eliminate $A$ by deleting sparse rows or sparse columns, and put each sparse row in $X$ and put each sparse column in $Y$. In particular, matrix multiplication $\BD \times \BD$ decomposes into operations of the form $\CS \times \CS$, $\RS \times \CS$, $\CS \times \RS$, and $\RS \times \RS$. In summary, we have
	\[
	\US \subseteq
	\left\{ \begin{aligned} \RS \\[-3pt] \CS \end{aligned} \right\}
	\subseteq \BD \subseteq \AS \subseteq \GM.
	\]
	
	As we will see, all of our algorithmic results are symmetric w.r.t.\ the three matrices. To help present such results, we will introduce the following shorthand notation. Let $\XX$, $\YY$, and $\ZZ$ be families of matrices. We write $[\XX:\YY:\ZZ]$ to refer to the following six operation: computing $\XX \times \YY = \ZZ$, $\YY \times \XX = \ZZ$, $\XX \times \ZZ = \YY$, $\ZZ \times \XX = \YY$, $\YY \times \ZZ = \XX$, and $\ZZ \times \YY = \XX$.
	
	Now equipped with this notation, we can rephrase the main result from the prior work \cite{gupta-2022-sparse}: they show that $[\US : \US : \US]$ can be solved in $O(d^{\expprevsr})$ rounds for semirings and in $O(d^{\expprevfield})$ rounds for fields. On the other hand, in \cref{thm:AS-AS-AS-hard} we argue why it is unlikely to extend this all the way to $[\AS : \AS : \AS]$. The key question is what happens between the two extremes, $\US$ and $\AS$? In particular, could we solve $[\BD:\BD:\BD]$ efficiently?

	In this work we present a near-complete classification of the complexity of matrix multiplication for all possible combinations of classes $\US$, $\BD$, $\AS$, and $\GM$. We show that almost all cases fall in one of these classes:
	\begin{enumerate}
		\item We have an upper bound of $O(d^{1.832})$ rounds by \cref{thm:us-us-as}. For example, $[\US : \US : \AS]$ falls in this class.
		\item We have a lower bound of $\Omega(\log n)$ rounds by \cref{thm:US-BD-BD-hard} and an upper bound of $O(d^2 + \log n)$ rounds by \cref{thm:US-AS-GM,thm:BD-AS-AS}. For example, $[\BD:\BD:\BD]$ falls in this class.
		\item We have a lower bound of $\Omega(\sqrt{n})$ rounds by \cref{thm:routing-hard}. For example, $[\BD:\BD:\GM]$ falls in this class.
		\item We have a conditional lower bound by \cref{thm:AS-AS-AS-hard}: a fast algorithm would imply major improvements in \emph{dense} matrix multiplication. For example, $[\AS:\AS:\AS]$ falls in this class.
	\end{enumerate}
	There is one outlier in our current classification: $[\US:\US:\GM]$ admits a trivial algorithm that runs in $O(d^4)$ rounds, but we do not know if it can be solved in $O(d^{1.832})$ rounds.

	Our results are summarized in \cref{tab:sum-sparse}. The table summarizes the results for the case of semirings, but analogous results hold for the case of fields---the only difference is that the exponent $\expnewsr$ is replaced with $\expnewfield$ and we have got $1 \le \lambda \le 2-2/\omega < \expdensefield$.
	
	\begin{table}[t]
		\centering
		\caption{Summary of results (for semirings). Here $1 \le \lambda \le 4/3$ is the infimum of exponents $\lambda$ such that dense matrix multiplication for semirings can be done in $O(n^\lambda)$ rounds. The lower bounds marked with $\dagger$ only hold for certain permutations of the matrix families.}\label{tab:sum-sparse}
		\begin{tabular}{@{}lllll@{}}
			\toprule
			Sparsity
			& \multicolumn{2}{l}{Upper bound}
			& \multicolumn{2}{l}{Lower bound}
			\\
			\midrule
			$[\US:\US:\US]$ & $O(d^{\expnewsr})$ & \cref{thm:us-us-as} & $\Omega(d^\lambda)$ & trivial \\
			\multicolumn{1}{@{}c}{$\dotsc$} && \\
			$[\US:\US:\AS]$ && \\
			\midrule
			$[\US:\US:\GM]$ & $O(d^4)$ & trivial & $\Omega(d^\lambda)$ & trivial \\
			\midrule
			$[\US:\BD:\BD]$ & $O(d^2 + \log n)$ & \cref{thm:US-AS-GM}& $\Omega(d^\lambda)$, $\Omega(\log n)$ & \cref{thm:US-BD-BD-hard} \\
			\multicolumn{1}{@{}c}{$\dotsc$} && \\
			$[\US:\AS:\GM]$ && \\
			\midrule
			$[\BD:\BD:\BD]$ & $O(d^2 + \log n)$ & \cref{thm:BD-AS-AS} & $\Omega(d^\lambda)$, $\Omega(\log n)$ & \cref{thm:US-BD-BD-hard} \\
			\multicolumn{1}{@{}c}{$\dotsc$} && \\
			$[\BD:\AS:\AS]$ && \\
			\midrule
			$[\US:\GM:\GM]$ &&& $\Omega(\sqrt{n})$ $^\dagger$ & \cref{thm:routing-hard} \\
			\multicolumn{1}{@{}c}{$\dotsc$} && \\
			$[\GM:\GM:\GM]$ && \\
			\midrule
			$[\BD:\BD:\GM]$ &&& $\Omega(\sqrt{n})$ $^\dagger$ & \cref{thm:routing-hard} \\
			\multicolumn{1}{@{}c}{$\dotsc$} && \\
			$[\GM:\GM:\GM]$ && \\
			\midrule
			$[\AS:\AS:\AS]$ &&& $\Omega(n^{(\lambda-1)/2})$ & \cref{thm:AS-AS-AS-hard} \\
			\multicolumn{1}{@{}c}{$\dotsc$} && \\
			$[\GM:\GM:\GM]$ && \\
			\bottomrule
		\end{tabular}
	\end{table}

	\subsection{Key conceptual message: role of bounded degeneracy}

	Our work highlights the role that the class of bounded-degeneracy matrices plays in the theory and practice of distributed matrix multiplication. While the class of average-sparse matrices is arguably more natural, it is also too broad to admit fast algorithms, even if we are in the supported setting (\cref{thm:AS-AS-AS-hard}). On the other hand, the class of uniformly-sparse matrices familiar from prior work can be far too restrictive for many applications. Our work suggests a promising new notion of sparsity in this context: matrices with bounded degeneracy. First, such matrices naturally arise especially in graph-theoretic application, and they also contain as special cases the important classes of row-sparse and column-sparse matrices. Second, our work shows that bounded-degeneracy matrices admit very efficient algorithms in the low-bandwidth setting (and as a corollary, also in any setting that is at least as strong as this model), at least if we know the structure of the matrices in advance (\cref{thm:BD-AS-AS}). We believe this can open a new research direction that studies the complexity of matrix multiplication across different distributed and parallel settings for bounded-degeneracy matrices.

	\subsection{Related work and applications}
	
	Our model of computing is the low-bandwidth model; this is closely related to the model known as \emph{node-capacitated clique} or \emph{node-congested clique} in the literature \cite{node-capacitated2019}---in the low-bandwidth model each computer can send and receive one message, while the node-capacitated clique model is usually defined so that each computer can send and receive $O(\log n)$ messages per round. Both the low-bandwidth model and the node-capacitated clique model can be interpreted as variants of the \emph{congested clique} model \cite{congest-clique2003}. Indeed, any algorithm that runs in $T(n)$ rounds in the congested clique model can be simulated in $nT(n)$ rounds in the low-bandwidth model, and for many problems (such as the dense matrix multiplication) such a simulation also results in the fastest known algorithms for the low-bandwidth model. Both the low-bandwidth model and the congested clique model can be interpreted as special cases of the classic \emph{bulk synchronous parallel model} \cite{BSP1990}.
	
	Matrix multiplication in this context has been primarily studied in the congested clique model \cite{alg-method-congest-cliq2019,le2016further,DBLP:conf/opodis/Censor-HillelLT18,censor2021fast}. However, as observed by \cite{gupta-2022-sparse}, the congested clique model is poorly-suited for the study of sparse matrix multiplication: one can only conclude that for sufficiently sparse matrices, the problem is solvable in $O(1)$ rounds, and one cannot explore more fine-grained differences between different algorithms.
	
	Similar to \cite{gupta-2022-sparse}, we work in the \emph{supported} version of the low-bandwidth model. In general, the supported model \cite{foerster2019preprocessing,foerster19preprocessing,schmid13local-sdn} refers to a setting in which the structure of the input is known in advance (we can do arbitrary \emph{preprocessing} based on the structure of the input), while the specific instance is revealed at run time. In the case of graph problems, we can see an unweighted graph $G$ in advance, while the adversary reveals a weighted version $G'$ of $G$ (a special case being a $0/1$-weighted or $1/\infty$-weighted graph, which can be interpreted as a subgraph of $G$). For matrix multiplication, we know the sparsity structure in advance (i.e., which elements are potentially nonzero), while the adversary reveals the concrete values at run time.
	
	Our logarithmic lower bounds build on the classic results in the CREW PRAM model by Cook, Dwork and Reischuk \cite{cook86} and Dietzfelbinger, Kutylowski and Reischuk \cite{dietzfelbinger1994}. As observed by e.g.\ Roughgarden, Vassilvitskii and Wang \cite{roughgarden18}, such lower bounds can be applied also in more modern settings.

	One of the main applications that we have in mind here is triangle detection, which has been extensively studied in the distributed setting \cite{dolev2012tri,izumi2017triangle,10.1145/3460900,10.1145/3446330,chang2019improved,korhonen2017deterministic,10.1145/3382734.3405742,gupta-2022-sparse}. If we can multiply matrices, we can also easily detect triangles in a graph. Moreover, $[\US:\US:\US]$ corresponds to triangle detection in a bounded-degree graph, while e.g.\ $[\AS:\AS:\AS]$ corresponds to triangle detection in a sparse graph.
	
	\subsection{Open questions for future work}
	
	In this work we initiated the study of different notions of sparsity in the context of low-bandwidth matrix multiplication. This work also gives rise to a number of questions for future work.
	
	First, our running times---$O(d^{\expnewsr})$ rounds for semirings and $O(d^{\expnewfield})$ for fields---are still far from the conditional lower bounds $\Omega(d^{\explbsr})$ and $\Omega(d^{\explbfield})$, respectively. Pushing the running time further down is a major challenge for future work. If we follow the two-phase approach of \cite{gupta-2022-sparse}, our present work essentially makes the second phase optimal; the main hope for future improvements now lies in the first phase of the algorithm.
	
	Our algorithms are designed for the supported model. Eliminating the knowledge of the support (i.e., knowledge of the sparsity structure) is a major challenge for future work.
	
	Finally, there are some gaps in \cref{tab:sum-sparse}. For example, we do not know if $[\US:\US:\GM]$ is strictly harder than $[\US:\US:\US]$, and also \cref{thm:routing-hard} do not cover all permutations of the matrix families. Finally, we do not know if $O(d^2 + \log n)$ can be further improved to, say, $O(d^{\expnewsr} + \log n)$.
	
	\section{Preliminaries}
	
	We work in the low-bandwidth model. There are $n$ computers. Initially each computer holds its own part of $A$ and $B$, and eventually each computer has to report its own part of $X$. When we study sparse matrices with at most $dn$ nonzeros, we assume that each computer holds at most $d$ elements, while when we study dense matrices, we assume that each computer holds at most $n$ elements.
	
	For our algorithms, it does not matter how the input and output is distributed among the computers---with an additional $O(d)$ time we can permute the input and output as appropriate. Our lower bounds hold for any fixed distribution of input and output that may also depend on the support.
	
	\subsection{Supported model and indicator matrices}
	
	When we study matrix multiplication in the supported model, we assume that we know in advance indicator matrices $\hat{A}$, $\hat{B}$, and $\hat{X}$ that encode the structure of our instance: $\hat{A}_{ij} = 0$ implies $A_{ij} = 0$, $\hat{B}_{jk} = 0$ implies $B_{jk} = 0$, and $\hat{X}_{ik} = 0$ indicates that we do not need to compute the value of $X_{ik}$. When we make assumptions on sparsity, our assumptions refer to the sparsity of $\hat{A}$, $\hat{B}$, and $\hat{X}$. For example, when we study $\US \times \BD = \AS$, we assume that $\hat{A} \in \US(d)$, $\hat{B} \in \BD(d)$, and $\hat{X} \in \AS(d)$. It then also follows that $A \in \US(d)$ and $B \in \BD(d)$.
	
	\subsection{Tripartite graph and triangles}
	
	It will be convenient to assume that our matrices are indexed with indices from three disjoint sets $I$, $J$, and $K$, each of size $n$. For example, elements of matrix $A$ are indexed with $A_{ij}$ where $i \in I$ and $j \in J$. Following \cite{gupta-2022-sparse}, we write $\TTall$ for the set of all \emph{triangles}, which are triples $\{i,j,k\}$ with $i \in I$, $j \in J$, and $k \in K$ such that $\hat{A}_{ij} \ne 0$, $\hat{B}_{jk} \ne 0$, and $\hat{X}_{ik} \ne 0$.
	
	For a collection of triangles $\TT$, we write $G(\TT)$ for the tripartite graph $G(\TT) = (V,E)$, where the set of nodes is $V = I \cup J \cup K$ and there is an edge $\{u,v\} \in E$ if there is a triangle $T \in \TTall$ with $\{u,v\} \subseteq T$. For a set of nodes $U \subseteq V$, we write $\TT[U] = \{ T \in \TT : T \subseteq U \}$ for the set of triangles induced by $U$.
	
	We say that we \emph{process} a triangle $\{i,j,k\}$ if we have added the product $A_{ij}B_{jk}$ to the sum $X_{ik}$. The key observation is that processing all triangles is exactly equivalent to computing all values of interest in the product matrix $X = AB$.
	
	\subsection{Clusters and clusterings}
	
	We say that $U \subseteq V$ is a \emph{cluster} if $U = I' \cup J' \cup K'$ for $I' \subseteq I$, $J' \subseteq J$, $K' \subseteq K$, $|I'| = d$, $|J'| = d$, and $|K'| = d$. We say that a collection of triangles $\PP$ is \emph{clustered} if there are disjoint clusters $U_1, \dotsc, U_k$ such that $\PP = \PP[U_1] \cup \dotsb \cup \PP[U_k]$.
	
	A clustered collection of triangles can be processed efficiently by applying an algorithm for dense matrix multiplication in parallel to each cluster \cite{gupta-2022-sparse}. By applying the strategy of \cite{alg-method-congest-cliq2019} and the latest value of $\omega$ from \cite{vassilevska-williams-2024-omega}, we obtain:
	\begin{lemma}\label{lem:clustered-solve}
		A clustered instance of matrix multiplication can be solved in $O(d^{4/3})$ rounds over semirings, and in $O(d^{1.156671})$ rounds over fields.
	\end{lemma}

	\section{Handling few triangles fast}
	
	We will start by a general result that will form the foundation for all of our upper-bound results. \cite[Lemma 5.2]{gupta-2022-sparse} observed that if we have a \emph{balanced} instance in which each node touches at most $\kappa$ triangles, we can solve matrix multiplication in $O(\kappa)$ rounds. However, the key challenge is what to do in \emph{unbalanced} instances: we have at most $\kappa n$ triangles, but some nodes can touch many more than $\kappa$ triangles. The following lemma shows how to handle also such instances efficiently:
	
	\begin{lemma}\label{lem:few-triangles}
		Let $\TT$ be a set of triangles with $|\TT| \le \kappa n$, and assume that for each pair of nodes $u,v \in V$ there are at most $m$ triangles $T \in \TT$ with $\{u,v\} \in T$. Assume that each computer initially holds at most $d$ elements of $A$ and at most $d$ elements of $B$, and it will need to hold at most $d$ elements of $X$. Then all triangles in $\TT$ can be processed (over fields or semirings) in $O(\kappa + d + \log m)$ rounds.
	\end{lemma}
	
	We emphasize that this strictly improves over \cite[Lemma 5.1]{gupta-2022-sparse} in two different ways: our result is applicable in a more general setting, and we avoid the factor-$2$ loss in the exponent of the running time.
	
	\subsection{High-level plan}
	
	Our high-level plan for proving \cref{lem:few-triangles} goes as follows:
	\begin{enumerate}
		\item We start with an arbitrary collection of $\kappa n$ triangles $\TT$ over the original set of $3n$ nodes $V = I \cup J \cup K$.
		\item We construct a new \emph{virtual} collection of $\kappa n$ triangles $\TT'$ over a new set of $O(n)$ virtual nodes $V'$.
		\item The new virtual collection of triangles is \emph{balanced} in the sense that each virtual node touches only $\kappa$ triangles.
		\item The virtual nodes are assigned to real computers so that each real computer is responsible for only $O(1)$ virtual nodes.
		\item We show how we can route original data to the virtual nodes, so that each virtual node only needs to do work for each triangle that it touches, and we can also aggregate the results back to the original nodes that need it.
	\end{enumerate}
	The most challenging part is the final step, routing.
	
	\subsection{Virtual instance}
	
	For each node $v \in V$, let $\TT(v)$ be the set of triangles in $T \in \TT$ with $v \in T$. Let $t(v) = |\TT(v)|$ and $\ell(v) = \lceil t(v)/\kappa \rceil$. Our set of virtual nodes is
	\begin{align*}
		V' &= I' \cup J' \cup K', \\
		I' &= \bigl\{ (i,x) \bigm| i \in I, x \in \{1,2,\dotsc,\ell(i)\} \bigr\}, \\
		J' &= \bigl\{ (k,x) \bigm| j \in J, x \in \{1,2,\dotsc,\ell(j)\} \bigr\}, \\
		K' &= \bigl\{ (j,x) \bigm| k \in K, x \in \{1,2,\dotsc,\ell(k)\} \bigr\}.
	\end{align*}
	The key observation is that the set of virtual nodes is not too large:
	\[
	|V'| = \sum_v \ell(v) \le \sum_v (t(v)/\kappa + 1) \le 3|\TT|/\kappa + n = 4n.
	\]
	Now we can construct a new balanced collection of triangles $\TT'$ over $V'$ such that each virtual node $v' \in V'$ touches at most $\kappa$ triangles: for the first $\kappa$ triangles in $\TT(v)$ we replace $v$ with the virtual node $(v,1)$, for the next $\kappa$ triangles we replace $v$ with $(v,2)$, etc.
	
	We assign the virtual nodes to real computers so that each computer is responsible for at most $4$ virtual nodes. Therefore, in what follows, we can assume that the virtual nodes are computational entities, and we can simulate their work in the real computer network with constant overhead.
	
	\subsection{Routing}
	
	Now let us see how we will process the triangles. Consider a triangle $\{i,j,k\} \in \TT$, and let its unique virtual copy be $\{i',j',k'\} \in \TT'$. Conceptually, we proceed as follows:
	\begin{enumerate}
		\item The real computer $p(i,j)$ that initially holds $A_{ij}$ transmits it to virtual computer $i'$.
		\item The real computer $p(j,k)$ that initially holds $B_{jk}$ transmits it to virtual computer $i'$.
		\item Virtual computer $i'$ computes the product $A_{ij} B_{jk}$ and transmits it to the real computer $p(i,k)$ that is responsible for storing $X_{ik}$.
	\end{enumerate}
	If we did this in a naive manner, it would be prohibitively expensive; a more careful routing scheme is needed; see \cref{fig:routing} for an illustration.
	
	\begin{figure*}
		\vspace*{1cm}
		\centering
		\includegraphics[page=2]{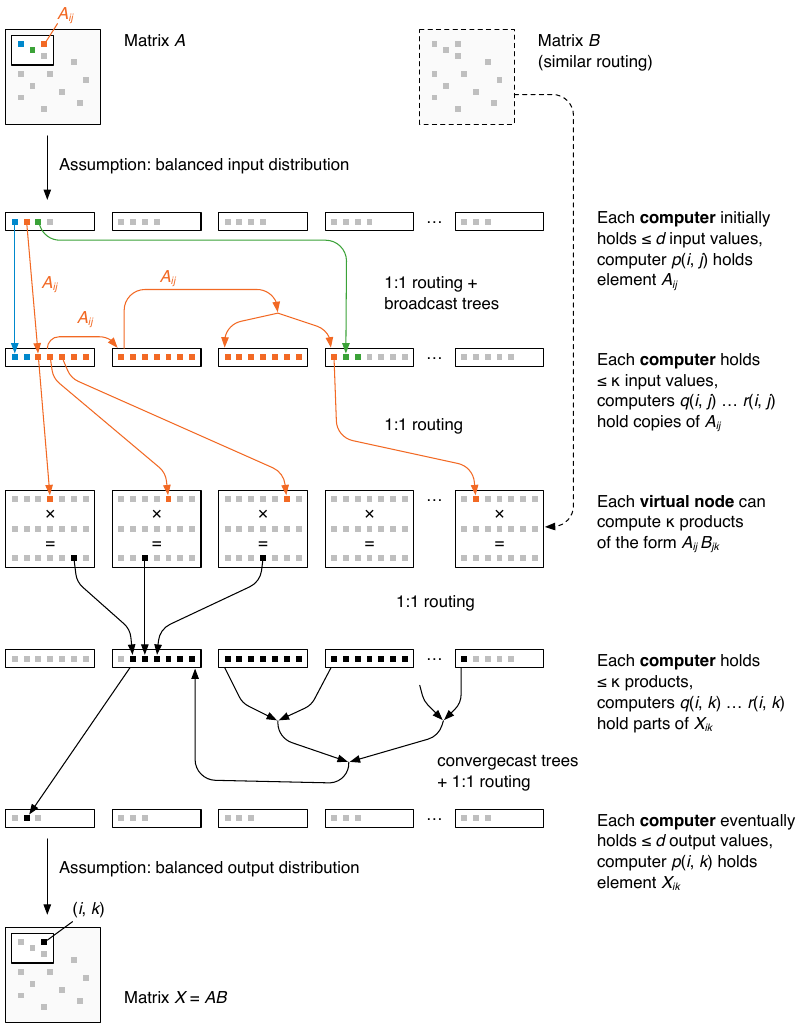}
		\vspace*{2mm}
		\caption{Routing scheme from the proof of \cref{lem:few-triangles}}\label{fig:routing}
		\vspace*{1cm}
	\end{figure*}
	
	First consider step (1): the transmission of $A_{ij}$ from $p = p(i,j)$ to $i'$. By assumption, ${i,j}$ is part of at most $m$ triangles, and a single real computer initially holds at most $d$ distinct values of matrix $A$. On the other hand, a single virtual computer needs to receive up to $\kappa$ messages, so the total number of messages to transmit is bounded by $\kappa n$.
	
	The key idea is that we will do intermediate-point routing. Form an array of all triples $(i,j,i')$ such that $A_{ij}$ is needed by $i'$, and order the array lexicographically $(i,j,i')$. For each $(i,j)$ the first triple $(i,j,i')$ is called the \emph{anchor} element for $(i,j)$. The array has at most $\kappa n$ triples, and we can assign the triples to our real computers so that each computer has at most $\kappa$ triples. For each $(i,j)$ we write $q(i,j)$ for the computer that holds the anchor element for $(i,j)$.
	
	Now we first route $A_{i,j}$ from computer $p(i,j)$ to computer $q(i,j)$. This can be completed in $O(d+\kappa)$ rounds, as each computer has $d$ outgoing messages and $\kappa$ incoming messages. We can implement this, for example, by considering the bipartite graph in which on one side we have sender, on one side we have recipients: the sender-side has maximum degree $d$ and the recipient-side has maximum degree $\kappa$, and hence we can find a proper edge coloring with $O(d+\kappa)$ colors and use the color classes to schedule the messages. This way in each round each computer sends and receives at most one message.
	
	Next, we will \emph{spread} the messages so that each computer that holds a triple $(i,j,i')$ knows the value of $A_{ij}$. Thanks to the array being sorted, this is now easy. Let us focus on some pair $(i,j)$. Recall that $q(i,j)$ is the first computer that holds any triple of the form $(i,j,\cdot)$, and it already knows the value of $A_{ij}$. Let $r(i,j)$ be the last computer that holds any triple of the form $(i,j,\cdot)$. Now our task is to simply spread $A_{ij}$ from $q(i,j)$ to all computers in the range $q(i,j)+1, \dotsc, r(i,j)$. If $r(i,j) = q(i,j)$, there is nothing to be done. On the other hand, if $r(i,j) \ne q(i,j)$, we have a convenient situation: computer $q(i,j)$ only needs to spread the value of $A_{ij}$, and computers $q(i,j)+1, \dotsc, r(i,j)$ only need to receive the value of $A_{ij}$. Hence, we can first let $q(i,j)$ inform $q(i,j)+1$, and then $q(i,j)+1, \dotsc, r(i,j)$ form a broadcast tree that distributes $A_{ij}$ to each of the computers. Note that all these broadcast trees involve disjoint sets of computers, and we can implement broadcast in parallel. We have at most $m$ triples of form $(i,j,\cdot)$, and hence (using a very sloppy estimate) at most $m$ computers involved in each broadcast operation. A broadcast tree of depth $O(\log m)$ suffices.
	
	At this point each real computer holds up to $\kappa$ triples of $(i,j,i')$ and knows the corresponding value $A_{ij}$. Then we route $A_{ij}$ to virtual computer $i'$; as everyone needs to send and receive up to $\kappa$ messages, this can be implemented in $O(\kappa)$ rounds.
	
	Step (2) is essentially identical to step (1); instead of values $A_{ij}$ and triples $(i,j,i')$, we transmit values $B_{jk}$ and use triples $(j,k,i')$.
	
	For step (3) we do the converse of step (1). We now form triples $(i,k,i')$ that indicate that virtual computer $i'$ holds a product $A_{ij} B_{jk}$ that needs to be accumulated to $X_{ik}$. We sort the array as above, and then first use $O(\kappa)$ rounds to route the products from virtual computers to real computers. Now each real computer that holds multiple triples of the form $(i,k,\cdot)$ can locally aggregate all these products into a single sum. Then we define $q(i,j)$ and $r(i,j)$ as above, and construct a convergecast tree with which computers $q(i,j)+1, \dotsc, r(i,j)$ compute the grand total of the products $A_{ij} B_{jk}$ they have received, and finally $q(i,j)+1$ relays this information to the anchor $q(i,j)$, which can compute the sum $X_{ij}$. Finally, we route $X_{ij}$ to the computer that needs to report it, using additional $O(\kappa + d)$ rounds.
	
	This completes the proof of \cref{lem:few-triangles}. In what follows, we will present applications of this result.
	
	\section{Algorithm for \texorpdfstring{\boldmath $[\US : \US : \AS]$}{[US : US : AS]}}
	
	In our terminology, \cite{gupta-2022-sparse} proved the following statement:
	
	\begin{theorem}[\cite{gupta-2022-sparse}]
		In the supported low-bandwidth model, sparse matrix multiplication of the form $[\US : \US : \US]$ can be computed in $O(d^{\expprevsr})$ rounds over semirings and in $O(d^{\expprevfield})$ rounds over fields.
	\end{theorem}
	
	In this section, we both improve the running time and widen the scope of applicability:
	
	\begin{theorem}\label{thm:us-us-as}
		In the supported low-bandwidth model, sparse matrix multiplication of the form $[\US : \US : \AS]$ can be computed in $O(d^{\expnewsr})$ rounds over semirings and in $O(d^{\expnewfield})$ rounds over fields.
	\end{theorem}
	
	\subsection{Preliminaries}
	
	We start with the following technical lemma:
	
	\begin{lemma}\label{lem:triangle-count}
		Let $\TTall$ be defined by an instance of $[\US : \US : \AS]$. For each node $x \in V$ there are at most $d^2$ triangles $T \in \TTall$ with $x \in T$.
	\end{lemma}
	\begin{proof}
		Consider first the case of $\US \times \AS = \US$. If $x \in I$, then $x$ is incident to at most $d$ nodes $j \in J$ and at most $d$ nodes $k \in K$, and hence in total there can be at most $d^2$ triangles of the form $\{x,j,k\}$. If $x \in J$, then $x$ is incident to at most $d$ nodes $i \in I$, and each of them is incident to at most $d$ nodes $k \in K$, and hence in total there can be at most $d^2$ triangles of the form $\{i,x,k\}$. The case of $x \in K$ is similar to $x \in J$. For the cases of $\AS \times \US = \US$ and $\US \times \US = \AS$, permute the roles of $I$, $J$, and $K$ as appropriate.
	\end{proof}
	
	We have two simple corollaries that will be useful shortly:
	
	\begin{corollary}\label{cor:triangle-edge}
		Let $\TTall$ be defined by an instance of $[\US : \US : \AS]$. For each pair of nodes $u,v \in V$ there are at most $d^2$ triangles $T \in \TTall$ with $\{u,v\} \in T$.
	\end{corollary}
	
	\begin{corollary}\label{cor:triangle-count}
		The total number of triangles in any $[\US : \US : \AS]$ instance is bounded by $d^2n$.
	\end{corollary}
	
	\subsection{Part 1: handling dense parts}\label{ssec:part1}
	
	In the first part we follow the basic structure of \cite{gupta-2022-sparse}, with two changes:
	\begin{enumerate}
		\item we generalize it to $[\US:\US:\AS]$,
		\item we choose different parameter values so that we can fully benefit from the new algorithm for part 2.
	\end{enumerate}
	We will then diverge from \cite{gupta-2022-sparse} in part 2, where we will apply \cref{lem:few-triangles}.
	
	Similar to \cite{gupta-2022-sparse}, we first show that if there are many triangles, we can find a \emph{dense cluster}:
	
	\begin{lemma}\label{lem:one-cluster}
		Let $\TTall$ be defined by an instance of $[\US : \US : \AS]$, and let $\TT \subseteq \TTall$. If $|\TT| \ge d^{2-\eps} n$ for some $\eps > 0$, then there exists a cluster $U \subseteq V$ with $|\TT[U]| \ge d^{3-4\eps}/24$.
	\end{lemma}
	\begin{proof}
		This is a generalization of \cite[Lemma 3.1]{gupta-2022-sparse}, which holds for $[\US : \US : \US]$. The original proof makes use of the following facts about $G(\TT)$:
		\begin{itemize}
			\item there are at most $dn$ edges between $J$ and $K$,
			\item each $i \in I$ is incident to at most $d$ nodes of $J$,
			\item each $i \in I$ is incident to at most $d$ nodes of $K$,
			\item each $j \in J$ is incident to at most $d$ nodes of $I$,
			\item each $k \in K$ is incident to at most $d$ nodes of $I$.
		\end{itemize}
		All of these also hold for $\US \times \AS = \US$. For the cases of $\AS \times \US = \US$ and $\US \times \US = \AS$, permute the roles of $I$, $J$, and $K$ as appropriate.
	\end{proof}
	
	Then by plugging in \cref{lem:one-cluster} instead of \cite[Lemma 3.1]{gupta-2022-sparse} in the proof of \cite[Lemma 4.1]{gupta-2022-sparse}, we obtain the following corollary:
	
	\begin{lemma}\label{lem:clustering}
		Let $\delta > 0$ and $\eps > 0$, and assume that $d$ is sufficiently large.
		Let $\TTall$ be defined by an instance of $[\US : \US : \AS]$, and let $\TT \subseteq \TTall$. Assume $|\TT| \ge d^{2-\eps}n$. Then we can partition $\TT$ into disjoint sets $\PP$ and $\TT'$ such that $\PP$ is clustered and $|\PP| \ge d^{2-5\eps-4\delta}n$.
	\end{lemma}
	\begin{proof}[Proof]
		We follow the strategy in the proof of \cite[Lemma 4.1]{gupta-2022-sparse}: Apply \cref{lem:one-cluster} to find a cluster $U$. Put the triangles $T \in \TT[U]$ that are fully contained in $U$ to $\PP$, and put the triangles $T \in \TT$ that only partially touch $U$ to $\TT'$. Repeat until there are sufficiently few triangles left; put all remaining triangles to $\TT'$.
		
		The original proof makes use of the fact that each node $x \in V$ is contained in at most $d^2$ triangles. This is trivial for $[\US : \US : \US]$. For $[\US : \US : \AS]$ we apply \cref{lem:triangle-count}, and then the original analysis holds verbatim.
	\end{proof}
	
	Finally, by plugging in \cref{lem:clustering} instead of \cite[Lemma 4.1]{gupta-2022-sparse} in the proof of \cite[Lemma 4.2]{gupta-2022-sparse}, we obtain the following corollary:
	
	\begin{lemma}\label{lem:many-clusterings}
		Let $\delta > 0$ and $0 \le \gamma < \eps$, and assume that $d$ is sufficiently large.
		Let $\TTall$ be defined by an instance of $[\US : \US : \AS]$, and let $\TT \subseteq \TTall$.
		Assume $|\TT| \le d^{2-\gamma}n$. Then we can partition $\TT$ into disjoint sets $\TT = \PP_1 \cup \dotsb \cup \PP_L \cup \TT'$ such that each $\PP_i$ is clustered, $L \le 144 d^{5\eps - \gamma + 4\delta}$ and $|\TT'| \le d^{2-\eps} n$.
	\end{lemma}
	\begin{proof}[Proof]
		We follow the proof of \cite[Lemma 4.2]{gupta-2022-sparse}: Apply \cref{lem:clustering} repeatedly to construct clusterings $\PP_1, \PP_2, \dotsc$. Stop once the residual set of triangles $\TT'$ is sufficiently small. The original analysis holds verbatim.
	\end{proof}
	
	Now we are ready to prove the following result. Note that unlike the analogous result from \cite{gupta-2022-sparse}, we have got here a convenient formulation in which the exponent in the time complexity matches the exponent in the size of the residual part:
	\begin{lemma}\label{lem:alg-part-1}
		Let $\TTall$ be defined by an instance of $[\US : \US : \AS]$. We can partition $\TTall$ into $\TT_1 \cup \TT_2$ such that:
		\begin{enumerate}
			\item For matrix multiplication over semirings, $\TT_1$ can be solved in $O(d^{\expnewsr})$ rounds and $|\TT_2| \le d^{\expnewsr} n$.
			\item For matrix multiplication over fields, $\TT_1$ can be solved in $O(d^{\expnewfield})$ rounds and $|\TT_2| \le d^{\expnewfield} n$.
		\end{enumerate}
	\end{lemma}
	\begin{proof}
		For the case of semirings, we apply \cref{lem:many-clusterings} with the parameters $\delta$, $\gamma$, and $\eps$ shown in \cref{tab:param-sr}. In the first step, we start with $\gamma = 0$; note that by \cref{cor:triangle-count} we satisfy the assumption that our initial set of triangles has size at most $d^{2-\gamma}n = d^2n$. We pick a very small $\delta$ and choose a suitable $\eps$. Then we use \cref{lem:clustered-solve} to process each of $\PP_i$. The total processing time will be $O(d^{\alpha})$, where $\alpha = 5\eps-\gamma+4\delta+4/3$. Our choice of $\eps$ is optimized so that $\alpha \le \expnewsr$ and hence we do not exceed our time budget. After processing all of $\PP_i$, we are left with the residual set of triangles of size $|\TT'| \le d^\beta n$, where $\beta = 2 - \eps$. We then repeat the same process, starting with $\gamma = 2-\beta$, and again optimizing $\eps$. We continue until $\beta \le \expnewsr$.
		
		For the case of fields, the idea is the same, but we make use of the parameter values in \cref{tab:param-field}.
	\end{proof}
	
	\begin{table}
		\centering
		\caption{Parameters for the proof of \cref{lem:alg-part-1} (semirings)}\label{tab:param-sr}
		\begin{tabular}{@{}l@{\qquad}ccccc@{}}
			\toprule
			Step & $\delta$ & $\gamma$ & $\eps$ & $\alpha$ & $\beta$ \\
			\midrule
			1 & 0.00001 & 0.00000 & 0.10672 & 1.86698 & 1.89328 \\
			2 & 0.00001 & 0.10672 & 0.12806 & 1.86696 & 1.87194 \\
			3 & 0.00001 & 0.12806 & 0.13233 & 1.86697 & 1.86767 \\
			4 & 0.00001 & 0.13233 & 0.13319 & 1.86700 & 1.86681 \\
			\bottomrule
		\end{tabular}
	\end{table}
	
	\begin{table}
		\centering
		\caption{Parameters for the proof of \cref{lem:alg-part-1} (fields)}\label{tab:param-field}
		\begin{tabular}{@{}l@{\qquad}ccccc@{}}
			\toprule
			Step & $\delta$ & $\gamma$ & $\eps$ & $\alpha$ & $\beta$ \\
			\midrule
			1 & 0.00001 & 0.00000 & 0.13505 & 1.83197 & 1.86495 \\
			2 & 0.00001 & 0.13505 & 0.16206 & 1.83197 & 1.83794 \\
			3 & 0.00001 & 0.16206 & 0.16746 & 1.83196 & 1.83254 \\
			4 & 0.00001 & 0.16746 & 0.16854 & 1.83196 & 1.83146 \\
			\bottomrule
		\end{tabular}
	\end{table}
	
	\subsection{Part 2: handling few triangles}\label{ssec:part2}
	
	Let us recap: We started with an arbitrary instance, with possibly up to $d^2 n$ triangles. Then we have used $O(d^{\alpha})$ rounds to process ``dense'' parts, and we are left with only $d^{\alpha} n$ triangles; here $\alpha = \expnewsr$ for semirings and $\alpha = \expnewfield$ for fields.
	
	Now we can apply \cref{lem:few-triangles} to handle all these remaining $d^{\alpha} n$ triangles in $O(d^{\alpha})$ rounds. When we apply the lemma, we set $\kappa = d^{\alpha}$, and by \cref{cor:triangle-edge} we have $m = d^2$. The overall running time is $O(d^{\alpha} + d + \log d^2) = O(d^\alpha)$. This completes the proof of \cref{thm:us-us-as}.

	\section{More general algorithms}
	
	We will now step beyond $[\US:\US:\AS]$ and consider more general settings. We will show that we can handle much more general notions of sparsity in $O(d^2 +\log n)$ rounds. For each case we prove a bound on the total number of triangles, and then apply \cref{lem:few-triangles}.
	
	\subsection{Algorithm for \texorpdfstring{\boldmath $[\US:\AS:\GM]$}{[US : AS : GM]}}
	
	\begin{lemma}\label{lem:US-AS-GM}
		Let $\TTall$ be defined by an instance of $[\US:\AS:\GM]$. Then $|\TTall| \le d^2n$.
	\end{lemma}
	\begin{proof}
		We give a proof for $\US \times \AS = \GM$; the remaining cases can be proved similarly.
		We can bound the number of triangles of the form $\{i,j,k\}$ with $i \in I$, $j \in J$, and $k \in K$ as follows: there are at most $dn$ edges of the form $\{j,k\}$, since $B \in \AS$. For each such edge $\{j,k\}$ there are at most $d$ edges of the form $\{i,j\}$, since $A \in \US$. It follows that the number of triangles is at most $d^2 n$.
	\end{proof}
	
	\begin{theorem}\label{thm:US-AS-GM}
		In the supported low-bandwidth model, sparse matrix multiplication of the form $[\US:\AS:\GM]$ can be computed in $O(d^2 + \log n)$ rounds over semirings and fields.
	\end{theorem}
	\begin{proof}
		Follows from \cref{lem:US-AS-GM} and \cref{lem:few-triangles} by setting $\kappa = d^2$ and $m = n$.
	\end{proof}
	
	\subsection{Algorithm for \texorpdfstring{\boldmath $[\BD:\AS:\AS]$}{[BD : AS : AS]}}
	
	When we consider $[\BD:\AS:\AS]$, it will be convenient to decompose $\BD$ into $\RS$ and $\CS$.
	
	\begin{lemma}\label{lem:RS-AS-AS}
		Let $\TTall$ be defined by an instance of $[\RS:\AS:\AS]$. Then $|\TTall| \le d^2n$.
	\end{lemma}
	\begin{proof}
		Let us consider the case $\RS \times \AS = \AS$; the other permutations are similar.
		We can bound the number of triangles of the form $\{i,j,k\}$ with $i \in I$, $j \in J$, and $k \in K$ as follows: there are at most $dn$ edges of the form $\{i,k\}$, since $X \in \AS$. For each such edge $\{i,k\}$ there are at most $d$ edges of the form $\{i,j\}$, since $A \in \RS$. It follows that the number of triangles is at most $d^2 n$.
	\end{proof}
	
	\begin{lemma}\label{lem:CS-AS-AS}
		Let $\TTall$ be defined by an instance of $[\CS:\AS:\AS]$. Then $|\TTall| \le d^2n$.
	\end{lemma}
	\begin{proof}
		Transpose the matrices in \cref{lem:RS-AS-AS}.
	\end{proof}
	
	\begin{lemma}\label{lem:BD-AS-AS}
		Let $\TTall$ be defined by an instance of $[\BD:\AS:\AS]$. Then $|\TTall| \le 2d^2n$.
	\end{lemma}
	\begin{proof}
		Let us consider the case $\BD \times \AS = \AS$; the other permutations are similar. As we discussed in \cref{ssec:intro-contrib-2}, we can decompose $X = AB$ into $X = A_1B + A_2B$, where $A_1 \in \RS$ and $A_2 \in \CS$. The claim follows by applying \cref{lem:RS-AS-AS} to $X_1 = A_1B$ and \cref{lem:CS-AS-AS} to $X_2 = A_2B$.
	\end{proof}
	
	\begin{theorem}\label{thm:BD-AS-AS}
		In the supported low-bandwidth model, sparse matrix multiplication of the form $[\BD:\AS:\AS]$ can be computed in $O(d^2 + \log n)$ rounds over semirings and fields.
	\end{theorem}
	\begin{proof}
		Follows from \cref{lem:BD-AS-AS} and \cref{lem:few-triangles} by setting $\kappa = 2d^2$ and $m = n$.
	\end{proof}
	
	\section{Lower bounds}
	
	Our lower bounds are organized as follows:
	\begin{itemize}
		\item \cref{ssec:lb-broadcast-aggregate} presents lower bounds that build on the hardness of broadcasting and aggregation. For example, an extreme case of $\BD \times \BD = \US$ is the multiplication of one dense row with one dense column. This enables us to compute arbitrary dot products, and in particular arbitrary sums. Hence matrix multiplication is at least as hard as computing the sum $n$ values. Similar arguments show that matrix multiplication is at least as hard as broadcasting one value to $n$ computers. We show that both of these tasks require $\Omega(\log n)$ rounds.
		\item \cref{ssec:lb-dense-mm} presents lower bounds that connect sparse matrix multiplication with dense matrix multiplication. If we have an algorithm for computing $\AS \times \AS = \AS$, we can pack a small dense matrix into one corner, and this way solve an arbitrary dense matrix multiplication. Particular care is needed here to account for the dual role that $n$ plays in our setting: it is both the matrix dimension and the number of computers.
		\item \cref{ssec:lb-routing} presents lower bounds that arise from communication complexity arguments. The key observation is that sparse matrix multiplication can be used to transmit information between two parties, Alice holding a subset of computers (and hence is able to manipulate their input) while Bob holds the rest of the computers (and hence is able to see their output). If there is a fast algorithm for matrix multiplication, Alice and Bob can simulate it to exchange information with each others too efficiently, leading to a contradiction.
	\end{itemize}
	Here it is good to note that the lower bounds in \cref{ssec:lb-broadcast-aggregate} are fundamentally related to the \emph{number} of messages (they hold even if the computers can send in each round one message that is arbitrarily large), while the lower bounds in \cref{ssec:lb-routing} are fundamentally related to the \emph{total size} of messages.

	\subsection{Broadcasting and aggregation}\label{ssec:lb-broadcast-aggregate}
	For settings between $[\US:\BD:\BD]$ and $[\US:\AS:\GM]$ our upper bound from \cref{thm:US-AS-GM} has an additive $O(\log n)$ term, which comes from broadcast and convergecast operations. In the following lemma, we show that some broadcasting and/or aggregation is indeed needed in all of these cases.
	
	\begin{lemma}\label{lem:US-BD-BD-hard}
		In the supported low-bandwidth model, sparse matrix multiplication of the form $[\US:\BD:\BD]$ is at least as hard as computing a sum of $n$ values distributed among $n$ computers, or broadcasting a single value to $n$ computers, even for $d = 1$.
	\end{lemma}
	\begin{proof}
		Consider first $\BD \times \BD = \US$. A special case of this is computing $AB = X$ such that all nonzeros in $A$ are in row $1$, all nonzeros in $B$ are in column $1$, and we are only interested in computing element $(1,1)$ of $X$. Furthermore, let all nonzeros of $B$ be $1$; we have the following task (where ``?'' indicates elements of the result matrix that we are not interested in):
		\[
		\begin{bmatrix}
			a_1 & a_2 & \dots & a_n \\
			0 & 0 & \dots & 0 \\
			\hdotsfor{4} \\
			0 & 0 & \dots & 0 \\
		\end{bmatrix}
		\times
		\begin{bmatrix}
			1 & 0 & \dots & 0 \\
			1 & 0 & \dots & 0 \\
			\hdotsfor{4} \\
			1 & 0 & \dots & 0 \\
		\end{bmatrix}
		=
		\begin{bmatrix}
			x & ? & \dots & ? \\
			? & ? & \dots & ? \\ 
			\hdotsfor{4} \\
			? & ? & \dots & ? \\ 
		\end{bmatrix}.
		\]
		Now this is equivalent to computing a sum $x = \sum_j a_j$, in a setting in which initially each computer holds one $a_j$ and one computer has to hold $x$.
		
		Then consider $\BD \times \US = \BD$. A special case of this is computing $AB = X$ such that all nonzeros in $A$ are in column $1$, the only nonzero of $B$ is at $(1,1)$, and we are only interested in computing elements in the first column of $X$. Furthermore, let all nonzeros of $A$ be $1$; we have the following task:
		\[
		\begin{bmatrix}
			1 & 0 & \dots & 0 \\
			1 & 0 & \dots & 0 \\
			\hdotsfor{4} \\
			1 & 0 & \dots & 0 \\
		\end{bmatrix}
		\times
		\begin{bmatrix}
			b & 0 & \dots & 0 \\
			0 & 0 & \dots & 0 \\
			\hdotsfor{4} \\
			0 & 0 & \dots & 0 \\
		\end{bmatrix}
		=
		\begin{bmatrix}
			x_1 & ? & \dots & ? \\
			x_2 & ? & \dots & ? \\ 
			\hdotsfor{4} \\
			x_4 & ? & \dots & ? \\ 
		\end{bmatrix}.
		\]
		Now we need to output $x_1 = \dotso = x_n = b$, and each computer has to hold one of these values. Hence, this is equivalent to broadcasting the value $b$ to all computers.
		
		The case of $\US \times \BD = \BD$ is also a broadcast task, by a similar argument.
	\end{proof}

	\subsubsection{Computing Boolean functions}\label{subsec:low_boolean_functions}

	In this section, we study the complexity of computing Boolean functions. In particular, we are interested in the hardness of computing the $\OR$ of $n$ bits distributed among $n$ computers ($\OR_n$ in short), as this would imply hardness for computing the sum of $n$ values.
	One major obstacle in proving lower bounds in this model is the possibility of communicating by silence. This section is largely based on ideas from \cite{dietzfelbinger1994} that handles similar issues in the CREW PRAM model.
	We first introduce a formal definition of the abstract low-bandwidth model, and provide preliminary information on Boolean functions before proving our main result of this section, namely, \Cref{thm:lowerbound_boolean}.

	\begin{definition}[abstract low-bandwidth model]
		
		\emph{Components}: A set of $n$ computers $C=\{c_1, c_2, \ldots, c_n \}$; a set $Q$ of states; an alphabet $\Sigma$; a single incoming port per computer that stores communicated data in $d_i \in \Sigma $; an output function $out: Q \rightarrow \Sigma$.
		
		Associated with each computer $c_i$ there are initial set of input $A_i \subseteq \Sigma$; an initial state $q_i^0 \in Q$; a state transition function $\delta_i: Q \times \Sigma \rightarrow Q$; a message function $\phi_i: Q \rightarrow \Sigma$; a message-address function $p_i:Q \rightarrow C$. These functions define an algorithm.
		
		\emph{Computation}: Computation proceeds in synchronous rounds, each round $t \in \{1, \ldots, T\}$ consists of two parts:
		\begin{enumerate}
			\item Local computation: each computer $c_i$ updates its state using the newly received data $d_i$ at round $t-1$, i.e.\ $q_i^t= \delta_i(q_i^{t-1}, d_i)$. Note that this admits unlimited local computation power. Moreover, $c_i$ determines message $\phi_i(q_i^t)$, and destination computer $p_i(q_i^t)$.
			\item Communication: each computer $c_i$ sends the message $\phi_i(q_i^t)$ to computer $c_j:=p_i(q_i^t)$, i.e.\ $d_j^{t+1}= \phi_i(q_i^t)$. It is required that computers send and receive at most one message. Note that it is permitted for computers to stay silent or not receive any messages, in this case, it is indicated by $p_i(q_i^t):=\Lambda$, and $d_j^{t+1}= \Lambda$ respectively, where $\Lambda$ a dedicated symbol in alphabet $\Sigma$.
		\end{enumerate}
		
	\end{definition}
	
	Note that the abstract version of the low-bandwidth model is stronger in the sense that there are no limits on the amount of data that a computer can communicate to another, for comparison, this is bounded to $O(\log n)$-bit words in our upper bounds. Together with unlimited local memory and computation, computers are able to send their entire state in a message. This makes our lower bounds quite strong.

	\subparagraph*{Boolean functions.} Given a Boolean function $f: \{0,1\}^n \rightarrow \{0,1\}$, and an input $a \in \{0,1\}^n$ evenly distributed between computers, we say that low-bandwidth algorithm $\alg$ computes $f$ in $T$ rounds if computer $c_1$ holds the value $\out(c_1)=f(a)$ by the end of round $T$. 
	
	Denote by $M_S$ the monomial $\Pi_{i \in S} x_i $ where $S \subseteq \{1, \ldots, n\}$. It is well known that any Boolean function $f$ can be written uniquely as a polynomial of form $\Sigma \alpha_S(f) M_S $ where $\alpha_S(f) \in \mathcal{R}$ has absolute value at most $2^{n-1}$ (see e.g., \cite{LECHNER1971}). We define \emph{degree} of a Boolean function $f$ as $\deg(f):= \{ \max |S| : \alpha_S(f) \neq 0\}$. Moreover, we denote by $\chi_S$ the characteristic function of $S \subseteq \{0,1\}^n$, and similarly define degree of a class of subsets $\mathcal{S}$ of $\{0,1\}^n$ to be $\deg(\mathcal{S}):= \max \{\deg(\chi_S) \mid S\in \mathcal{S} \} $.
	The following lemma gives a basic tool for calculating the degree of various combinations of Boolean functions.
	
	\begin{lemma}[Lemma 2.3 \cite{dietzfelbinger1994}] \label{lem:deg}
		For $f,g: \{0,1\}^n \rightarrow \{0,1\}$ the following hold:
		\begin{enumerate}[(a)]
			\item $\deg(f \wedge g)= \deg(f \cdot g) \leq \deg(f) + \deg(g) $.
			\item $ \deg(\bar{f}) = \deg(1-f) = \deg(f)$.
			\item $ \deg(f \vee g) = \deg(1-(1-f)(1-g)) = \deg(f+g-f\cdot g) \leq \deg(f) + \deg(g)$.
			\item If $f \wedge g \equiv O$, then $\deg(f \vee g) = \deg(f+g) \leq \max\{\deg(f), \deg(g) \}$.
			\item $\deg(f \wedge \bar{g}) = \deg(f \cdot (1-g)) \leq \deg(f) + \deg(g)$.
		\end{enumerate}
	\end{lemma}
	
	\begin{lemma}[analogous to Theorem 3 in \cite{dietzfelbinger1994}] \label{thm:lowerbound_boolean}
		Computing $f: \{0,1\}^n \rightarrow \{0,1\}$ in the supported low-bandwidth requires $\Omega(\log \deg(f))$ rounds.
	\end{lemma}

	Before proving \Cref{thm:lowerbound_boolean}, we introduce some useful definitions.
	\begin{definition}
		For $t\le T$, state $q\in Q$, and symbol $s\in \Sigma$ we define
		\begin{align*}
			G(q,c,t)&:=\{a\in \{0,1\}^n : \text{ computer $c$ is in state $q$ after $t$ rounds on input $a$}  \}, \\
			\GG(t)&:=\{ G(q,c,t): q \in Q, c\in C \}.
		\end{align*}
	\end{definition}
	
	For any input $a\in G(q,c,t)$, computer $c$ will be in state $q$ in round $t$, intuitively, this means that computer $c$ only knows that the input is in set $G(q,c,t)$. In fact, one can show that classes of these partitions are essentially the states. 
	Now, we have all the tools necessary to prove the hardness result.

	\begin{proof}[Proof of \Cref{thm:lowerbound_boolean}]
		Note that $f^{-1}(1)$ is the disjoint union of all $G(q,c_1,T)$ where $\out_1(q)=1$. Therefore, it is enough to show that $\deg(\GG(T)) \leq 2^T$.
		We show that the following hold:
		\begin{enumerate}[(a)]
			\item $\deg(\GG(0))=1$. 
			\item $\deg(\GG(t)) \leq 2 \deg(\GG(t-1))$, for $t>0$.
			\item $\deg(\GG(t)) \leq 2^t$, for $t>0$.
		\end{enumerate}
		In order to prove (a), note that the initial state of computer $c_i$ depends on the input $a_i\in \{0,1\}$. Let $G:=G(q,c_i,0) \in \GG(0)$, then $\chi_G$ is either $x_i$ or $1-x_i$ based on input bit $a_i$. Hence, the degree of the characteristic function $\deg(\GG(0))= \max \{  \deg(G) : G \in \GG(0) \} =1$.
		
		Next, we show that (b) holds, then (c) immediately follows. Let $a\in G=G(q',c_i,t), t\geq 1$, be an input. There are two possible cases:
		
		\emph{Case 1: On input $a$, some computer $c_j\in C$ sends a message to $c_i$.}
		Let $q_i$ and $q_j$ be the states of computers $c_i$ and $c_j$ in round $t-1$ of computation on $a$. Let $G_i=G(q_i,c_i,t-1)$ and $G_j=G(q_j,c_j,t-1)$.  It is clear that $G= G_i \cap G_j$. By \Cref{lem:deg}(a), $\deg(\chi_G)= \deg(\chi_{G_i \cap G_j})=  \deg(\chi_{G_i} \cdot \chi_{G_j}) \leq \deg(\chi_{G_i}) + \deg(\chi_{G_j}) \leq 2 \deg(\GG(t-1))$.

		\emph{Case 2: On input $a$, computer $c_i$ receives no messages at the end of round $t-1$.}
		Let $G_1, \ldots, G_m$ be the list of all computer partitions where $G_k=G(q,c_{i_k},t-1)$ for $k\in \{1,\ldots, m\}$ that computer $c_{i_k}$ would send a message to $c_i$ in round $t-1$. The key observation is that $G_1, \ldots, G_m$ are disjoint. For the sake of contradiction let us assume that $G_k \cap G_j \neq \emptyset$. Therefore, there is an input $\hat{a} \in G_k \cap G_j$ such that both computers $c_{i_k}$ and $c_{i_j}$ send a message to $c_i$ in round $t-1$, a contradiction of the communication rule. Let $G^*:= \bigcup_{j=1}^{m} G_j$ and let $G^{t-1}= G(q,c_i,t-1)$ for any $q\in Q$. Note that the new partition at round $t$ is $G=G^{t-1} \cap G^*$.
		Applying the disjointness property together with \Cref{lem:deg}(d) shows that $\deg(G^*) \leq \deg(\GG(t-1))$. Thus, $\deg(G) \leq 2 \deg(\GG(t-1))$.
	\end{proof}

	The following are consequences of \Cref{thm:lowerbound_boolean}.
	
	\begin{corollary}\label{cor:OR_lowerbound}
		In the supported low-bandwidth model, computing $\OR$ of $n$ bits distributed among $n$ computers takes $\Omega{(\log n)}$ rounds.
	\end{corollary}
	\begin{proof}
		Follows from \Cref{thm:lowerbound_boolean} and the fact $\deg(\OR_n)=n$. Note that $\OR_n(x)$ can be written as polynomial $P(x)=1-\Pi_{i=1}^{n} (1-x_i)$.
	\end{proof}
	
	\begin{corollary}\label{cor:sum-hard}
		In the supported low-bandwidth model, computing the sum of $n$ values distributed among $n$ computers takes $\Omega{(\log n)}$ rounds.
	\end{corollary}
	\begin{proof}
		Existence of an $o(\log n)$-round algorithm that computes the sum would imply an $o(\log n)$-round algorithm for computing $\OR_n$, which contradicts \Cref{cor:OR_lowerbound}.  
	\end{proof}

	\subsubsection{Broadcast}\label{subsec:broadcast}
	
	Next, we show that broadcasting even a single bit $b$ to all computers in the supported low-bandwidth model takes $\Omega{(\log n)}$ rounds. We assume that $b\in \{0,1\}$ is given to computer $c_1$ as input and this information (not the value of $b$) is available to all computers as part of the support.
	Let $\bb(c) \in \{0, 1, \bot \}$ be the internal state of computer $c$ on what the broadcast value is, where $\bot$ means undecided. Initially, $\bb(c_1)=b$ and the rest of computers are undecided, i.e.\ $\bb(c)= \bot, c\neq c_1$. 
	The goal is to have $\bb(c)=b$ for all computers.
	
	Let $p_i(c,b')$ be the destination (computer) to which computer $c$ sends a message at round $i$ given that $\bb(c)=b'$.
	
	\begin{definition}
		We say that computer $c$ \emph{affects} computer $c'$ in round $i$ if all the following conditions hold: (i) $\bb_{i-1}(c')=\bot$, (ii) $\bb_{i}(c')=\{0,1\}$, and (iii) $p_i(c,s)=c'$ where $s \in \{0,1, \bot \}$.
	\end{definition}
	
	\begin{lemma}\label{lem:low_broadcast}
		In the supported low-bandwidth model, broadcasting a single bit $b \in \{0,1\}$ to all computers requires $\Omega{(\log n)}$ rounds.
	\end{lemma}
	\begin{proof}
		Note that an undecided computer $c$ cannot affect $p_i(c, \bot)$, since this would imply $\bb_{i-1}(c)=b$. Therefore, in order to broadcast $b$, all computers must be affected at some point.
		Let \[B_i:=| \left\{c : \bb(c)\neq \bot \text{ in round }i\right\}|\] be the number of affected computers in round $i$. For the sake of simplicity, we assume that $c_1$ is initially affected (e.g.\ by input). In any $T$-round broadcast algorithm $\alg$, we have that $B_1=1$ and $B_T=n$. In round $i$, each affected computer $c$ can affect at most $2$ other computers $p_i(c, \bb(c))$ and $p_i(c, 1-\bb(c))$. Note that the latter is affected by not receiving a message. Hence,
		\begin{align*}
			B_1&=1, \\
			B_T&=n, \\
			B_i &\leq B_{i-1} &&\text{ (already affected) }  \\
			&+ B_{i-1} &&\text{ (affected by communication) }  \\
			&+ B_{i-1} &&\text{ (affected by silence) } \\
			&= 3B_{i-1}.
		\end{align*}
		We have $n = B_T \leq 3^{T}$ and $T \geq \log_3 n$.
	\end{proof}
	
	\subsubsection{Wrap up}

	Now we have all the building blocks for our lower-bound result:
	\begin{theorem}\label{thm:US-BD-BD-hard}
		In the supported low-bandwidth model, sparse matrix multiplication of the form $[\US:\BD:\BD]$ requires $\Omega(\log n)$ rounds.
	\end{theorem}
	\begin{proof}
		Follows from \cref{lem:US-BD-BD-hard,cor:sum-hard,lem:low_broadcast}.
	\end{proof}
	
	\subsection{Dense matrix multiplication}\label{ssec:lb-dense-mm}
	
	Let us now argue why $[\AS:\AS:\AS]$ is unlikely to be solvable in $O(d^2 + \log n)$ rounds, or even in $f(d) + O(n^{\eps})$ rounds for sufficiently small values of $\eps$.
	
	\begin{lemma}\label{lem:AS-AS-AS-hard}
		Computing $[\AS:\AS:\AS]$ for $d = 1$ in $T(n)$ rounds implies an algorithm that solves dense matrix multiplication in $T'(n) = n T(n^2)$ rounds.
	\end{lemma}
	\begin{proof}
		Assume $\alg$ solves $[\AS:\AS:\AS]$ in $T(n)$ rounds for $d = 1$. Let $m = \sqrt{n}$. We use $\alg$ as a black box to design an algorithm $\alg'$ that multiplies \emph{dense} $m \times m$ matrices using $m$ computers in $T'(m)$ rounds. Algorithm $\alg'$ works as follows: given an input $X' = A'B'$ with dimensions $m \times m$, construct an instance $X = AB$ of dimensions $n \times n$ by padding with zeros:
		\[
		\begin{bmatrix}
			A' & 0 \\
			0 & 0 \\
		\end{bmatrix}
		\times
		\begin{bmatrix}
			B' & 0 \\
			0 & 0 \\
		\end{bmatrix}
		=
		\begin{bmatrix}
			X' & ? \\
			? & ? \\
		\end{bmatrix}.
		\]
		This is now an average-sparse instance with $m^2 = n$ nonzeros. Then apply algorithm $\alg$, so that each of the $m$ computers simulates $m = \sqrt{n}$ computers; this way we have $n$ virtual computers available for running $\alg$. The running time of $\alg$ is $T(n) = T(m^2)$, and hence the simulation completes in time $T'(m) = m T(m^2)$.
	\end{proof}

	\begin{theorem}\label{thm:AS-AS-AS-hard}
		Assume that we can compute $[\AS:\AS:\AS]$ for $d = 1$ in $o(n^{(\lambda-1)/2})$ rounds in the supported low-bandwidth model, for some $\lambda > 1$. Then we can also solve dense matrix multiplication in $o(n^\lambda)$ rounds.
	\end{theorem}
	\begin{proof}
		By plugging in $T(n) = o(n^{(\lambda-1)/2})$ in \cref{lem:AS-AS-AS-hard}, we obtain $T'(n) = n \cdot o(n^{\lambda-1})$.
	\end{proof}
	Note that in the case of semirings, there has been no progress with dense matrix multiplication algorithms that are faster than $\Omega(n^{4/3})$. Hence by plugging in $\lambda = 4/3$ in \cref{thm:AS-AS-AS-hard}, we conjecture that $[\AS:\AS:\AS]$ for semirings requires at least $\Omega(n^{1/6})$ rounds.
	
	\subsection{Routing}\label{ssec:lb-routing}
	
	Next we consider problem settings that are at least as hard as routing many elements to a single computer. We will assume here that the assignment of the input and output values to computers only depends on the structure of the input, and not on the numerical values. We will write $A^v$, $B^v$, and $X^v$ to denote the set of elements of matrices $A$, $B$, and $X$ that are held by computer $v$.
	
	\subsubsection{Routing with \texorpdfstring{\boldmath $[\US:\GM:\GM]$}{[US:GM:GM]}}

	We first consider $[\US:\GM:\GM]$. We prove a hardness result for $\US \times \GM = \GM$; the case of $\GM \times \US = \GM$ is symmetric, but $\GM \times \GM = \US$ is left for future work.
	
	\begin{lemma}\label{lem:US-GM-GM-hard}
		In the supported low-bandwidth model, to solve $\US \times \GM = \GM$, at least one computer needs to output $\Omega(\sqrt{n})$ values originally held by other computers.
	\end{lemma}
	
	\begin{proof}
		For the sake of simplicity, we give a proof when matrix $A$ contains at most $2n$ nonzero elements, i.e.\ $d=2$. Let all the elements of $A$ except $a_{i,i}, a_{i,(i \text{ mod }n)+1}$ for all $i \in [n]$ are equals to $0$; that is, our task is to compute
		\[
		\begin{bmatrix}
			a_{1,1} & a_{1,2} & 0 & 0 & \dots & 0 & 0 \\
			0 & a_{2,2} & a_{2,3} & 0 & \dots & 0 & 0 \\
			\hdotsfor{7} \\
			a_{n,1} & 0 & 0 & 0 & \dots & 0 & a_{n,n} \\
		\end{bmatrix}
		\times
		B
		=
		X,
		\]
		where $B$ is a general (dense) matrix, and we are interested in all values of $X$.
		
		Each computer in the network contains two (nonzero) elements of $A$, $n$ elements of $B$ and $n$ elements of $X$. We divide the analysis into the following two cases; consider a computer $v$:
		\begin{enumerate}
			\item If $X^v$ contains at least $\sqrt{n}$ elements from one of the columns of $X$: let us assume that $X^v$ contains $\sqrt{n}$ elements from column $j$ of $X$ i.e.,  \[x_{i_1,j},\ x_{i_2,j},\ \ldots,\ x_{i_{\sqrt{n}},j} \in X^v.\] We know that \[x_{i,j} = a_{i,i} b_{i,j} + a_{i,(i \text{ mod }n)+1} b_{(i \text{ mod }n)+1,j}.\] Let us set $ b_{i,j} =1,$ and $a_{i,(i \text{ mod }n)+1} = 0$ for all $i,j \in [n]$. This implies $x_{i_l,j} = a_{i_l,i_l}$ for all $i_l \in [\sqrt{n}]$. Since computer $v$ initially holds only two of these $a_{i_l,i_l}$, we need to route at least $\Omega(\sqrt{n})$ values from other computers to computer $v$.
			
			\item If $X^v$ contains less than $\sqrt{n}$ elements from every column of $X$: let \[I =\{(i,j) \mid x_{i,j} \in X^v\}.\] If  we set  $a_{i,i} =1$ and $a_{i,((i-1) \text{ mod }n)+1} = 0$ for all $i \in [n]$, we have $X= B$. Thus, $v$ must output all the elements of $b_{i,j}$ such that $(i,j) \in I$. Therefore, $v$ must receive all these elements.        
			Now it might happen that these $n$ elements of $B$ are already stored at $v$. However, we will prove that computer $v$ also requires at least $\sqrt{n}$ more elements stored at other computers in order to produce correct output for different values of nonzero elements of $A$ and $B$. Note that if $X^v$ contains less than $\sqrt{n}$ nonzero elements from each column of $X$, then the number of columns of $X$ from which $X^v$ contains an element is more than $\sqrt{n}$. Let $C^v$ be the set of columns of $X$ such that $X^v$ contains some element from these columns. From above, we can say that in each $C_i \in C^v$ there exist two elements $x_{j,i},x_{j+1,i}$ such that $x_{j,i} \in X^v$ and $x_{j+1,i} \notin X^v$ (because $X^v$ contains less than $\sqrt{n}$ elements from each column). Let \[J= \{j \mid x_{j,i} \in X^v \text{ and } x_{j+1,i} \notin X^v \}.\] Now we take a different assignment where we keep matrix $B$ the same, and in matrix $A$ we set $a_{j,j} = 0$ and $a_{j,(j \text{ mod }n +1)} = 1$ for all $j \in J$. This will result in $x_{j,i} = b_{j+1,i}$ for all $i$ and $j$ such that $x_{j,i} \in X^v$ and $x_{j+1,i} \notin X^v$. We know that there exist more than $\sqrt{n}$ such $x_{j,i}$. From the above paragraph, we know that $v$ must contain all $b_{j,i}$ such that $x_{j,i} \in X^v$. Therefore, we can say that $v$ does not contain $b_{j+1,i}$ such that $x_{j,i} \in X^v$ and $x_{j+1,i} \notin X^v$. Thus, $v$ needs to receive $\Omega(\sqrt{n})$ values from other computers.\qedhere
		\end{enumerate}
	\end{proof}

	\subsubsection{Routing with \texorpdfstring{\boldmath $[\BD:\BD:\GM]$}{[BD:BD:GM]}}

	Next we consider $[\BD:\BD:\GM]$. We prove a hardness result for $\BD \times \BD = \GM$; the case of $\BD \times \GM = \BD$ and $\GM \times \BD = \BD$ is left for future work. We consider the special case of $\RS \times \CS = \GM$; the case of $\BD \times \BD = \GM$ is at least as hard:
	
	\begin{lemma}\label{lem:RS-CS-GM-hard}
		In the supported low-bandwidth model, to solve $\RS \times \CS = \GM$, at least one computer needs to output $\Omega(\sqrt{n})$ values originally held by other computers.
	\end{lemma}
	
	\begin{proof}
		We prove the lower bound for the case when $d=1$, i.e.\ each computer in the network contains one element of $A$, one element of $B$ and $n$ elements of $X$. Let $A$ and $B$ be matrices such that all the elements of these matrices are zero except $a_{i,1}$ and $b_{1,i}$ for all $i \in [n]$:
		\[
		\begin{bmatrix}
			a_{1,1} & 0 & \dots & 0 \\
			a_{2,1} & 0 & \dots & 0 \\
			\hdotsfor{4} \\
			a_{n,1} & 0 & \dots & 0 \\
		\end{bmatrix}
		\times
		\begin{bmatrix}
			b_{i,1} & b_{1,2} & \dots & b_{1,n} \\
			0 & 0 & \dots & 0 \\
			\hdotsfor{4} \\
			0 & 0 & \dots & 0 \\
		\end{bmatrix}
		=
		X.
		\]
		Similar to the above, we divide the analysis into two cases.
		\begin{enumerate}
			\item If $X^v$ contains at least $\sqrt{n}$ elements from one of the columns of $X$: let us assume that $X^v$ contains $\sqrt{n}$ elements from column $j$ of $X$ i.e.,  \[x_{i_1,j},\ x_{i_2,j},\ \ldots ,\ x_{i_{\sqrt{n}},j} \in X^v.\] Let us set $b_{1,i} = 1$ for all $i \in n$. This implies that $x_{i_k, j}$ = $a_{i_k,1}$ for all $k \in \sqrt{n}$. We know that $v$ contains only one element of $A$, and therefore we need to receive $\Omega(\sqrt{n})$ elements from other computers.
			\item If $X^v$ contains less than $\sqrt{n}$ elements from each column of $X$: in such cases we know that there exist at least $\sqrt{n}$ elements $x_{i_1,j_1}, x_{i_2,j_2}, \ldots x_{i_{\sqrt{n}},j_{\sqrt{n}}}$ such that $j_i\neq j_k$ for $i \neq k$ and $i,k \in [\sqrt{n}]$. Now consider an assignment where we set all $a_{i,1} = 1$. This implies $x_{i_k,j_k} = b_{1,j_k}$ for all $k \in \sqrt{n}$. Again we can conclude that we need to receive $\Omega(\sqrt{n})$ elements from other computers.\qedhere
		\end{enumerate}
	\end{proof}
	
	\subsubsection{Communication complexity}

	Let us consider the following communication task between Alice and Bob: Alice has a $k$-length vector $V$ that Bob wants to output (or learn). Each element in the vector can be represented by $\log n$ bits. In each round Bob can send and receive only $\log n$ bits (similar to the low-bandwidth model).
	
	\begin{lemma}
		\label{clm:Alice-Bob}
		Any algorithm will require at least $k$ rounds for Bob to learn vector $V$.
	\end{lemma}
	\begin{proof}
		Assume that there is an algorithm such that Bob outputs $V$ in $t<k$ rounds. Let $A$ be a communication vector that Bob receives, that is, Bob receive $i^{th}$ element of $A$ in $i^{th}$ round. Let $\mathcal{A}$ be the set of all such possible vectors $A$ and $\mathcal{V}$ be the set of all possible vectors $V$. Notice that Bob's outputs depends on the communication vector it receives. Therefore, we can say that there is map from $\mathcal{V} \rightarrow \mathcal{A}$ such that $V$ is mapped to $A$ if for vector $V$ that Alice has, Bob receives the communication vector $A$. We know that $|\mathcal{A}| = 2^{t \log n}| < |\mathcal{V}| = 2^{k \log n}$. Therefore, there exist two different vectors $V_1,V_2 \in \mathcal{V}$ that are mapped to the same vector $A' \in \mathcal{A}$. Thus, If we use this algorithm then for two different vectors $V_1$ and $V_2$ held by Alice, Bob will output the same vector. Which is a contradiction.
	\end{proof}
	
	We now interpret \cref{lem:US-GM-GM-hard,lem:RS-CS-GM-hard} as communication tasks so that the computer that outputs $\Omega(\sqrt{n})$ values represents Bob and all other computers collectively represent Alice. Note that Alice and Bob can simulate the communication network so that each of them only needs to send and receive $O(\log n)$ bits. Applying \cref{clm:Alice-Bob} we get the following result:
	\begin{theorem}\label{thm:routing-hard}
		Any algorithm in the supported low-bandwidth model requires $\Omega(\sqrt{n})$ rounds to compute $\US \times \GM = \GM$ or $\RS \times \CS = \GM$.
	\end{theorem}
	
	\bibliographystyle{plainurl}
	\bibliography{triangles}
	
\end{document}